\documentclass[copyright,creativecommons]{eptcs}
\providecommand{\event}{LINEARITY 2014} 
\usepackage{breakurl}             
\usepackage{amsmath}
\usepackage{bussproofs}
\usepackage{stmaryrd}
\usepackage{cmll}
\usepackage{amssymb}
\usepackage{latexsym}
\usepackage{graphicx}
\usepackage{xcolor}
\usepackage{proof}
\usepackage{wasysym}\usepackage{hyperref}
\usepackage{amssymb}
\usepackage{amsthm}
\usepackage{tikz}
\allowdisplaybreaks

\providecommand{\doi}[1]{doi: \href{http://dx.doi.org/#1}{\nolinkurl{#1}}}

\def\eg{e.g., }
\def\cf{cf.  }
\def\ie{i.e., }

\def\wrt{w.r.t. }
\renewcommand{\b}{{}^{\bot}}

\newcommand{\ifff}{if and only if }

\newcommand{\dom}{\textnormal{dom}}

\newcommand{\pow}{\textnormal{pow}}

\newcommand{\sP}{\mathsf{P}}

\newcommand{\sz}{\mathsf{Z}}

\newcommand{\eqdef}{\stackrel{\text{\tiny \textnormal{DEF}}}{=}}
\newcommand{\reldef}{\stackrel{\text{\tiny \textnormal{DEF}}}{\Longleftrightarrow}}

\newcommand{\lb}{\big\{ \; }
\newcommand{\rb}{\; \big\}}

\newcommand{\st}{ \ \ \big{|} \ \ }

\newcommand{\bA}{\mathbb{A}}

\newcommand{\ort}{\bot}

\newcommand{\cA}{\mathcal{A}}

\newcommand{\cF}{\mathcal{F}}

\newcommand{\cN}{\mathcal{N}}
\newcommand{\cP}{\mathcal{P}}

\newcommand{\sO}{\mathsf{O}}
\newcommand{\sr}{\mathsf{R}}

\newcommand{\ssf}{\mathsf{p}}
\newcommand{\ssg}{\mathsf{o}}

\newcommand{\spo}{\, \lhd_\cO\,}

\newcommand{\sco}{\, \LHD_\cO\,}
\newcommand{\scd}{\, \LHD_{\overline{\cO}} \; }

\newcommand{\cO}{\mathcal{O}}

\newcommand{\orth}[2]{\, #1\, \bot \, #2 \,}

\newcommand{\arr}{\longrightarrow}

\newtheorem{theorem}{Theorem}[section]   
\newtheorem{lemma}[theorem]{Lemma}
\newtheorem{proposition}[theorem]{Proposition}
\newtheorem{corollary}[theorem]{Corollary}
\theoremstyle{definition}
\newtheorem{definition}[theorem]{Definition}

\newtheorem{example}[theorem]{Example}

{\bfseries}{}

\newcommand{\squishlist}{
 \begin{list}{$\bullet$}
  { \setlength{\itemsep}{0pt}
     \setlength{\parsep}{3pt}
     \setlength{\topsep}{3pt}
     \setlength{\partopsep}{0pt}
     \setlength{\leftmargin}{1em}
     \setlength{\labelwidth}{1.5em}
     \setlength{\labelsep}{0.5em} } }
\newcommand{\squishend}{
  \end{list}  }

\newcommand{\squishlistt}{
 \begin{list}{$\bullet$}
  { \setlength{\itemsep}{0pt}
     \setlength{\parsep}{3pt}
     \setlength{\topsep}{3pt}
     \setlength{\partopsep}{0pt}
     \setlength{\leftmargin}{1em}
     \setlength{\labelwidth}{1.5em}
     \setlength{\labelsep}{0.3em} } }
\newcommand{\squishendd}{
  \end{list}  }

\newcommand{\sqi}{
 \begin{list}{$\bullet$}
  { \setlength{\itemsep}{0pt}
     \setlength{\parsep}{3pt}
     \setlength{\topsep}{3pt}
     \setlength{\partopsep}{0pt}
     \setlength{\leftmargin}{1.4em}
     \setlength{\labelwidth}{1.5em}
     \setlength{\labelsep}{0.3em} } }
\newcommand{\sqe}{
  \end{list}  }

  \newcommand{\sqqi}{
 \begin{list}{$\bullet$}
  { \setlength{\itemsep}{0pt}
     \setlength{\parsep}{3pt}
     \setlength{\topsep}{3pt}
     \setlength{\partopsep}{0pt}
     \setlength{\leftmargin}{1em}
     \setlength{\labelwidth}{1.5em}
     \setlength{\labelsep}{0.3em} } }
\newcommand{\sqqe}{
  \end{list}  }

  \newcommand{\sqii}{
 \begin{list}{}
  { \setlength{\itemsep}{0pt}
     \setlength{\parsep}{3pt}
     \setlength{\topsep}{3pt}
     \setlength{\partopsep}{3pt}
     \setlength{\leftmargin}{3.2em}
     \setlength{\labelwidth}{3.8em}
     \setlength{\labelsep}{0.6em} } }
\newcommand{\sqee}{
  \end{list}  }

\title{Ludics without Designs I: Triads}
\author{Michele Basaldella\thanks{Supported by the ANR project ANR-2010-BLAN-021301 LOGOI.}
\institute{Universit\'e d'Aix--Marseille, CNRS, I2M, Marseille, France}
\email{michele.basaldella@gmail.com}}
\def\titlerunning{Ludics without Designs I: Triads}
\def\authorrunning{Michele Basaldella}
\begin{document}
\maketitle

\begin{abstract}
In this paper, we introduce the  concept of triad.  Using this notion, we  study,  revisit, discover and rediscover some basic properties of ludics
from a very general point of view.
\end{abstract}

\section{Introduction}

An orthodox introduction of a paper on
ludics should begin as follows.
 First, the authors  say what ludics is commonly intended to be:
typically, they would say that it is a kind of game semantics
which is close to the more popular categorical game models for linear logic and $PCF$   introduced in the last twenty years.
Having set up the context,
then they could informally describe ludics as an untyped and monistic framework which provides a  semantics for  proofs of a linear (\ie without exponentials) polarized fragment of linear logic. The authors should also stress that ludics is a  semantics which is based on interaction. Finally,
|  trying not to  frighten the  casual reader |
the authors should  give an   intuitive  account of some of the basic constituents  of ludics: the notions of design,
 orthogonality, behaviour, etc., putting
  more emphasis on the concepts
which are more related to the contribution of the paper.

Of course, there is \emph{nothing}  wrong (or bad)   in starting an article on ludics in the ``orthodox way" described above. However, for this paper we find more instructive to take
another approach. Namely, we give from the very beginning the most important definition of our work.

\begin{definition}[Triad] \label{defss}
A \textbf{triad} is an ordered triple
 $ A= (\cP_A,\cN_A,\ort_A)$ where:
\squishlist
\item $\cP_A = \lb p,q,r, \ldots \rb $ is a set. Its  elements    are said to be
\textbf{positive terms}.
\item $\cN_A = \lb n,m,\ell , \ldots \rb$ is a  set. Its elements  are said to be
\textbf{negative terms}.
\item The sets $\cP_A$ and $\cN_A$ are \emph{disjoint}.  We call the set $\cP_A \cup \cN_A$  the \textbf{domain of $A$} and we denote it as $\dom (A)$.  Elements of $\dom (A)$ are also said to be \textbf{terms}.

\item $\ort_A $ is a relation $\bot_A \subseteq \cP_A \times \cN_A$ called
\textbf{orthogonality}.
For $p \in \cP_A$ and $n \in \cN_A$, we write
$\, p \, \bot_A \, n \, $
and $\, p \, \not\!\!\bot_A \, n \, $ for $(p,n) \in \bot_A$ and
 $(p,n) \notin \bot_A$, respectively. \hfill $\triangle$
\squishend
\end{definition}

We now give a simple example of triad.

\begin{example} \label{ex triad}
Let  $I \eqdef (\cP_I,\cN_I,\ort_I)$ be the ordered triple given as follows.
\sqi
\item  Let
$P$ and $N$ be two  distinct symbols.
\item
 Let $\cP_I \eqdef\lb 0,1,2 \rb \times \lb P\rb $ and
  $\cN_I \eqdef \lb 0,1,2 \rb \times \lb N \rb$.
  Clearly, $\cP_I$ and $\cN_I$  are  disjoint sets.
The domain of $I$ is, of course, the set 
$ \lb (0,P), (1,P), (2,P),  (0,N), (1,N), (2,N)  \rb$.
 \item
 $\bot_I$ is  given as follows:

\smallskip
 {\centering
$(r,P) \, \bot_I \, (s,N)  \  \  \reldef \ \  r = s$\enspace, \qquad for  $(r,P) \in \cP_I$ and  $(s,N) \in \cN_I$\enspace.  \par}
\smallskip
  \sqe
The triple $I$ is a triad in our sense.  \hfill $\triangle$
\end{example}
The notion of triad is not at all a new mathematical concept. Except
for minor details, similar structures have  already been defined and investigated in the literature of several fields of research.
To be short, we only mention the notion of
 \emph{context}
in \emph{formal concept analysis}
(see \eg \cite{bookB, Zhang})
and the notion of  \emph{classification domain} (or \emph{classification}) in \emph{information theory} (see \eg \cite{Barwise, bookC}).

In  the field of   research
which concerns this paper | \ie the
\emph{proof theory}  related to \emph{linear logic} |  we remark that
our notion of triad   is very similar
to the concept of \emph{Boolean--valued  game} \cite{DBLP:conf/lics/LafontS91} or  \emph{Boolean--valued Chu space}  (see \eg \cite{DBLP:conf/lics/Pratt95, Zhang}), as we discuss
 in more detail in  Section \ref{example triads}.

 How  is the notion of triad  related to ludics?

As the title of this work should suggest
(``ludics without designs" is indeed quite provocative; it sounds like ``proof theory without proofs")
here we do not consider  a \emph{design} | the very central concept of ludics | as the \emph{well--specified} and concrete  proof--like
object  defined in \cite{DBLP:journals/mscs/Girard01} (see also \cite{DBLP:journals/corr/abs-cs-0501039,DBLP:journals/tcs/Terui11} for other more or less equivalent definitions of design).
Rather,  a design
 is seen as special case of  what we are calling \emph{term}:   just an
\emph{unspecified} and primitive  element of the domain  of a given triad  $A$.

Similarly, in ludics there is a  \emph{well--specified}
orthogonality relation \cite{DBLP:journals/mscs/Girard01}: the one   which relate two elements $p$ and $n$  if and only if the  procedure of normalization between the designs $p$ and $n$ successfully terminates.
By contrast, here  we consider a more general situation: we are interested in all  possible orthogonality relations. Given two disjoint sets
$\cP_A$ and $\cN_A$,  \emph{any} subset of
$\cP_A \times \cN_A$ is an orthogonality in our sense. In particular, we do not need to recall or introduce any kind of procedure of normalization.

  What can we do with triads?

\squishlist
\item In Subsection \ref{sub clo} we use the orthogonality relation $\bot_A$ to define \emph{closed sets} (in the sense of \cite{closures}) in $\cP_A$
and $\cN_A$,
and we study some  basic properties.
Closed sets   are called 
  \emph{behaviours} in ludics  \cite{DBLP:journals/mscs/Girard01},
and they are the semantical counter--part of the syntactical notion of formula in logic.

\item In Subsection \ref{spec sub} we introduce  the \emph{specialization relation}
on $\cP_A$ and $\cN_A$. Our relation of specialization
is exactly  the \emph{precedence relation} between designs   defined in
\cite{DBLP:journals/mscs/Girard01}.
Furthermore, we generalize
 specialization
to a new relation that we call \emph{semantical consequence} and study some of its properties.
\item
In Subsection \ref{ent sys sec} we
introduce the notion of \emph{entailment system}.
This notion can be seen as the natural adaptation of the concept
 \emph{information system} \cite{Scott} (see also \cite{Barwise,Zhang})
to our setting. We show that  triads equipped by
 relations of semantical consequence
are entailment systems.
This  result is useful to us because
it allows us to understand   properties of
the relation of
semantical consequence for triads in terms of  ``structural  rules" of entailment systems.
\squishend

Next, in Section \ref{functionals} we introduce and study the concept of \emph{functional} for a triad.
In ludics,
functionals are introduced in  \cite{DBLP:journals/entcs/BasaldellaST10},
and they are the designs which constitute | categorically speaking |  the \emph{morphisms} in ludics. In view of their importance, one of the aims of this paper is to study some fundamental properties of functionals at a more general and abstract level.

\squishlist
\item In Subsection \ref{fun subsec}
we generalize the notion of functional  introduced in 
 \cite{DBLP:journals/entcs/BasaldellaST10}
to our setting.

\item In Subsection \ref{sub prop fun}
we define the notion of \emph{continuous} functional
and the notion of functional
which \emph{preserves the relation
of semantical consequence}.
We show that  these two notions
are equivalent.

\item In Subsection \ref{lin sec}
we define the notion of \emph{regular} functional. This notion  is crucial in our work because in ludics every functional  (in the sense of \cite{DBLP:journals/entcs/BasaldellaST10}) is  regular.
Regularity 	has also very pleasant consequences:
 \emph{regular functionals are continuous and   preserve the relation of semantical consequence}.
\squishend

In Section  \ref{example triads} we give a couple of
examples of triads and functionals. In particular,
we show that designs and functionals
as given in  \cite{DBLP:journals/entcs/BasaldellaST10} meet our conditions.
In Section \ref{conclusions} we conclude.

Our methodology is the following: except for Section \ref{example triads}, we always
work with \emph{arbitrary triads}.
This (obviously) means that the results we are going to show
hold in any triad, and | more importantly |  that these results hold
in ludics  without explicitly introducing
the notion of  design nor the specific orthogonality relation of ludics.
The reader should be able to understand our abstract results of Section \ref{closed sets} and Section \ref{functionals}  without any previous  knowledge of ludics: all we need  is in  Definition \ref{defss}.

To conclude this section,
we would like point out that
some of the  results we show  in this paper are  perhaps not  new, as  structures similar to our triads have
been studied extensively in the literature.
On the other hand,  we also remark that |  to the best of our knowledge   |
many  constructions we are considering in this paper
and, consequently,  many  results  stated in  Section \ref{closed sets} and Section \ref{functionals}
 seem to be new, 
 \emph{when concretely applied to ludics}. 
More specifically, 
we mainly refer to: 
\sqi
\item
the construction  of the relation of semantical consequence | which allows us to understand  closed sets (\ie behaviours) in terms
 of sets of  consequences of  entailment
systems;
\item the results on functionals concerning  the notions of continuity and regularity |
which  allow us to better understand the nature of the functionals of ludics (in the sense of \cite{DBLP:journals/entcs/BasaldellaST10}).
\sqe
We collect the most significant results in Theorem \ref{thm fin}.


\section{Triads: Basic Theory}
\label{closed sets}

In this section, we study some basic property of triads.
In Subsection
\ref{sub clo} we introduce the notion of closed set in our setting.
In Subsection \ref{spec sub}
 we introduce  the relations of specialization and semantical consequence, respectively, and study some properties. Finally, in Subsection
\ref{ent sys sec} we
introduce the notion of entailment system and relate this concept
with the relation of semantical consequence.


 We now
 fix some notation and terminology.
Let $ A= (\cP_A,\cN_A,\ort_A)$
be a triad.
\sqi
\item[(1)]   We use $a,b,c,\ldots$ to range over terms (\ie over elements of the domain of $A$).

\item[(2)]  We use the letter $\cO_A$
  as a variable ranging over $\lb \cP_A, \cN_A \rb$.
Furthermore,  if $\cO_A = \cP_A$, then we write
$\overline{\cO_A}$ for $\cN_A$. Similarly, if
 $\cO_A = \cN_A$, then we write
$\overline{\cO_A}$ for $\cP_A$.
Note that $\overline{\overline{\cO_A}} = \cO_A$.

\item[(3)]  By an abuse of notation,
for $a \in \cO_A$ and $b \in \overline{\cO_A}$ we write   $a \, \ort_A \,b\, $,   or equivalently $b \, \ort_A \, a \,$,  for

\smallskip

{\centering
the positive term in $ \lb a,b \rb$  \ \  $\bot_A$ \ \  the negative term in $\lb a,b \rb$\enspace.
\par}

\smallskip

\noindent This notation makes sense precisely because the sets $\cO_A$ and $\overline{\cO_A}$ are supposed to be disjoint.
\item[(4)]  We write $\emptyset_{\cO_A}$ to mean that  the empty--set
$\emptyset$ has to be intended as a subset of $\cO_A$.
\item[(5)]  Given a set $E$ we write
$\pow(E)$ for the power--set of $E$ (\ie the set of its subsets).
\item[(6)]  We use the expression ``iff'' as an abbreviation for ``if and only if.'' 
\sqe

From now on, up to the end the paper, we fix an arbitrary triad
$ A= (\cP_A,\cN_A,\ort_A)$
and an arbitrary subset of terms
$\cO_A \in \lb \cP_A , \cN_A \rb$. To ease notation, in the rest of the paper
 we  write  $\cP$, $\cN$, $\bot$ and $\cO$  for
 $\cP_A$, $\cN_A$, $\bot_A$ and  $\cO_A$, respectively.

\subsection{Closed Sets} \label{sub clo}

In this subsection, we  use the orthogonality relation $\bot$ to equip the sets $\cP$ and $\cN$
with some topological structure. Namely,   we introduce the concept of \emph{closed set} in our setting.
Here, closed sets are not to be intended as ``closed sets in a \emph{topological space}" but as
``closed sets in a \emph{closure space} \cite{closures}", a slightly more general topological notion.
Closed sets are important to us because they correspond to  \emph{behaviours}  \cite{DBLP:journals/mscs/Girard01,DBLP:journals/corr/abs-cs-0501039,DBLP:journals/tcs/Terui11},
the ``semantical" notion of formula in ludics.
Closed sets induced by orthogonality are not uncommon in the literature of theoretical computer science, see \eg
\cite{DBLP:conf/popl/VouillonM04,DBLP:journals/tcs/Paolini08}.

\begin{definition}[Orthogonal sets, closed sets] \label{orto} Let $X \subseteq \cO$.
 We define  the \textbf{orthogonal set of $X$}
as the subset $X^\bot $ of  $ \overline{\cO}$ given by:
\smallskip

{\centering
$ b \in X\b \ \reldef \  \orth{a}{b}$ \ for every $ a \in X$\enspace, \qquad for \,\!\! $b \in\overline{\cO}$\enspace.
\par}

\smallskip

\noindent
We call \textbf{closed set in $\cO$} any set $X \subseteq \cO$    such that $X = X\b\b$.
  \hfill $\triangle$
\end{definition}

The following theorem  justifies the terminology given in Definition \ref{orto}.

\begin{theorem}[Closure] \label{closed}  Let $X$ and $Y$ be subsets of  $\cO$, and let $
a \in \cO$ and $b \in \overline{\cO}$.  Then, we have:

\smallskip

{\centering
$
\begin{array}{rlcrl}
\!\!\hbox{\emph{(1)}} \!\!&\!\!  X \subseteq X \b\b \!\ ;
 & \phantom{as} &
 \!\!\!\!\!\!\!\! \!\!\hbox{\emph{(2)}} \!\!&\!\! X \subseteq Y  \hbox { \!\ implies \!\ }  Y\b \subseteq X \b  \!\ ;  \\ \!\!\hbox{\emph{(3)}} \!\!&\!\! X \subseteq Y  \hbox { \!\ implies \!\ }  X\b\b \subseteq Y\b\b \!\ ;  & &
  \!\!\!\!\!\!\!\!\!\!\hbox{\emph{(4)}} \!\!&\!\! X\b =  X \b\b\b \!\ ; \\
 \!\!\hbox{\emph{(5)}} \!\!&\!\!  X\b\b = X\b\b\b\b  \!\ ; & \phantom{a} &
  \!\!\!\!\!\!\!\!\!\!\hbox{\emph{(6)}} \!\!&\!\!   \mbox{$X\b$ is a closed set in $\overline{\cO}$} \!\ ;  \\

 \!\!\hbox{\emph{(7)}} \!\!&\!\!   \mbox{$X\b\b$ is a closed set in $\cO$} \!\ ; & \phantom{as} &
    \!\!\!\!\!\!\!\!\!\!\hbox{\emph{(8)}} \!\!&\!\!  \mbox{$X$ is a  closed set in $\cO$ just in case $X= Z\b\b$ for some $Z \in \cO$} \!\ ; \\

     \!\!\hbox{\emph{(9)}} \!\!&\!\!   \mbox{$\orth{a}{b}$ just in case  $a \in \lb b \rb \b$} \!\ ; & \phantom{as} &
      \!\!\!\!\!\!\!\!\!\!\!\!\!\!\hbox{\emph{(10)}} \!\!&\!\!  \mbox{$\orth{a}{b}$ just in case $b \in \lb a \rb \b$} \!\ . \\

  \end{array}
$

\par}

\vspace{-0.2cm}

\end{theorem}
\begin{proof}
(1) : \   Let $c \in X$.  We  have  $\orth{c}{d}$ for every $d \in X\b$  by definition of $X\b$.  Hence, $c \in X\b\b$.

(2) : \  Let $c \in Y\b$. Then,  $\orth{c}{d}$ for every $d \in Y$.
As $X \subseteq Y$, we have $\orth{c}{d}$ for every $d \in X$.
So, $ c \in X\b$.

(3) : \   By  applying (2) above  two times, we get the result.

(4) : \  By (1) above,  we have $X\b \subseteq (X\b)\b\b = X\b\b\b$. By (1) again, we obtain
$X \subseteq X\b\b$. By (2) above, we get $X\b\b\b =(X\b\b)\b  \subseteq X\b$. Hence, we have $X\b = X\b\b\b$.

(5) : \    By (4) above, we have $X\b = X\b\b\b$. Hence,  $X\b\b = (X\b)\b = (X\b\b\b)\b = X\b\b\b\b$.

(6) and (7) : \  By (4) and (5) above, we have $X\b = (X\b)\b\b$ and
 $X\b\b = (X\b\b)\b\b$, respectively.

(8) : \  Suppose that $X =Z\b\b$ for some $Z \subseteq \cO$. Then, we have
$X\b\b =  Z\b\b\b\b = Z\b\b =X$ by using (5) above. As for the converse,
suppose that $X$ is a  closed in $\cO$, \ie that $X = X\b\b$ holds.
Take $Z \eqdef X\b\b$. Then,  we have
$Z\b\b = X\b\b\b\b = X\b\b = X$ by using (5) above again.

(9)  and
(10) : \   They immediately follow from the definition of orthogonal set.
\end{proof}

We note that properties (1), (3) and (5) above say that $\b\b$ is a closure operator on the sets $\cP$ and $\cN$, \ie that $(\cP,\b\b)$
and $(\cN,\b\b)$ form two \emph{closure spaces}, in the sense of \cite{closures}.  However, we  want to point out that our theory is richer than the theory of closure spaces:  the latter can be ``axiomatized" by using
properties (1), (3) and (5) above and the  double--orthogonality operator $\b\b$ only; in particular,
 there is nothing there which  corresponds to property (2) above, for instance.

In the sequel, we frequently use the properties
listed in  Theorem \ref{closed} without  any explicit reference.



\begin{example} \label{ex triad cont} Let $I$ be the triad given in Example \ref{ex triad}.
We  calculate the closed sets
in $\cP_I$ and $\cN_I$.
\sqi
\item $(\emptyset_{\cP_I})^{\bot_I \bot_I}= (\cN_I)^{\bot_I}= \emptyset_{\cP_I}$.
\item  $\lb (r,P) \rb^{\bot_I \bot_I} = \lb (r,N) \rb^{\bot_I}= \lb (r,P) \rb$, for every $r \in \lb 0,1,2 \rb$.

\item   $X^{\bot_I \bot_I} = (\emptyset_{\cN_I})^{\bot_I}= \cP_I$, for any other subset $X$ of $\cP_I$.
\item $ (\emptyset_{\cN_I})^{\bot_I \bot_I} = (\cP_I)^{\bot_I} = \emptyset_{\cN_I}$.
\item   $\lb (s,N) \rb^{\bot_I \bot_I} = \lb (s,P) \rb ^{\bot_I}= \lb (s,N) \rb$, for every $s \in \lb 0,1,2 \rb$.

\item   $X^{\bot_I \bot_I} = (\emptyset_{\cP_I})^{\bot_I}= \cN_I$, for any other subset $X$ of $\cN_I$. \hfill $\triangle$
\sqe
\end{example}

\subsection{The Relation of Specialization  and the Relation of Semantical Consequence} \label{spec sub} \label{sen sec}

We now introduce the \emph{specialization relation} in our setting. Specialization is a   relation which  has been extensively studied in topology, order theory and domain theory.
In our setting, we define the specialization relation as follows.

\begin{definition}[The specialization relation $\spo$]  \label{specrel}
We define the  \textbf{specialization relation} as the binary relation  $\spo \, \subseteq \cO \times \cO$ given by:
\vspace{-0.3cm}
\begin{equation}\tag*{$\triangle$}
\mbox{$ a \spo b \ \reldef \  \lb a \rb\b \subseteq \lb b \rb\b\enspace, \qquad \hbox{for \,\!\! $a$  \,\!\! and  \,\!\! $b$\, \!\!  in \,\!\!  $\cO$}\enspace.$ }
\end{equation}
\end{definition}

In the sequel, we  read the expression $a \spo b$ as ``$a$ is more special than $b$.''

Our definition of $\spo$  follows  the analogous relation defined in \cite{DBLP:journals/mscs/Girard01}.
There, the specialization relation is called \emph{precedence} relation, and in \cite{DBLP:journals/corr/abs-cs-0501039,DBLP:journals/tcs/Terui11}
it is called \emph{observational} ordering (but note, however, that there the observational ordering  is defined in a different manner). In ludics, we can read $a \spo b$  as ``$a$ is more defined than $b$" (see \cite{DBLP:journals/mscs/Girard01}).

We also  point out that several authors define
the specialization relation as the
inverse  of  $\spo$. For instance, in \cite{closures}  specialization for closure spaces is defined as in Proposition \ref{proppsec}(iii) below, but with the role of $a$ and $b$ interchanged.




\begin{example} \label{ex triad contI} Let $I$ be the triad given in Example \ref{ex triad}. We  have:
\sqi
\item 
$ (r,P) \, \lhd_{\cP_I} \, (r',P)\,$ \ifff $r=r'$, for every $r$ and $r'$ in $\lb 0,1,2\rb$.

\item Analogously,
$ (s,N) \, \lhd_{\cN_I} \, (s',N)\,$ \ifff $s=s'$, for every $s$ and $s'$ in $\lb 0,1,2\rb$. \hfill $\triangle$
\sqe
\end{example}

The following proposition gives us a useful characterization of the relation of specialization.

\begin{proposition} \label{proppsec}
Let $a$ and $b$ in $\cO$. Then, the  following  claims are equivalent:

\smallskip

{\centering

$
\begin{array}{rlcrlcrl}
\hbox{\emph{(i)}} &  a \spo b \!\ ;
 & \phantom{as} &
\hbox{\emph{(ii)}}  & b \in \lb a\rb\b\b \!\ ; &  \phantom{as} &
\hbox{\emph{(iii)}}  & \mbox{For every $X \subseteq \cO$, \ $a \in X\b\b$ \ implies \  $b \in X\b\b$\enspace.}  \end{array}
$

\par}

\vspace{-0.2cm}
\end{proposition}

\begin{proof}
 (i) implies (iii) : \  Let $X \subseteq \cO$. Suppose that $a \in X\b\b$, \ie $\lb a\rb \subseteq X\b\b$.
Then, we have $ X\b =X\b\b\b \subseteq \lb a \rb\b$. Suppose that
 $a \spo b$, \ie
$\lb a \rb\b \subseteq \lb b \rb\b$. Then,  we have
$X\b  \subseteq \lb b \rb\b$.
Thus,
 $ \lb b \rb\b\b \subseteq   X\b\b$.
Therefore, $ b \in \lb b \rb \subseteq \lb b \rb\b\b \subseteq X\b\b$.

 (iii) implies (ii) : \ Let $X = \lb a \rb$. We have
 $ a \in  \lb a \rb  \subseteq \lb a \rb\b\b$. So,
$ b \in  \lb a \rb\b\b$  by (iii).

 (ii) implies (i) : \
Assume $ b \in \lb a \rb\b\b$, \ie  $ \lb b \rb  \subseteq \lb a \rb\b\b$.
Then, we have
$ \lb 	a \rb\b = \lb a \rb\b\b\b  \subseteq \lb b \rb\b$.  Hence, $a \spo b$.
\end{proof}

\begin{corollary} \label{coro sp}
 For every $a$ and $b$ in $\cO$, we have  \,\!\!
$ b \in \lb a \rb\b\b $ \ if and only if \ $ a \spo b$ \!\,.
\end{corollary}
\begin{proof}
It immediately follows from the equivalence of properties
(i) and (ii) of Proposition \ref{proppsec}.
\end{proof}

By   Corollary
 \ref{coro sp}, for any singleton subset
$\lb a \rb$ of $\cO$,
the  closed set
$\lb a\rb \b\b$ can be completely
described by using specialization: the members
of  $\lb a\rb \b\b$
are exactly the terms $b \in \cO$
such that $a \spo b$ holds.
Our next step
is to generalize this kind of property  to arbitrary sets,
\ie not only singletons.
To do this, we need to generalize the relation of specialization. This motivates the following definition.

\begin{definition}[The semantical consequence relation $\sco$] \label{sem cons}
We define the  relation of \textbf{semantical consequence} as the binary relation  $\sco \ \subseteq \pow(\cO) \times \cO$ given by:

\vspace{-0.3cm}
\begin{equation}\tag*{$\triangle$}
\mbox{$ X \sco b \ \ \reldef \ \   \bigcap_{a \in X}\lb a \rb\b \subseteq \lb b \rb\b\enspace, \qquad \hbox{for  \,\!\! $X \subseteq \cO$  \,\!\! and   \,\!\! $b \in \cO$}\enspace.$ }
\end{equation}
\end{definition}

In the sequel, we  read the expression $X \sco b$ as ``$b$ is a semantical consequence of $X$.''

 Regarding our terminology, we call the relation
$\sco$ \emph{semantical consequence} because so it is called in similar contexts (\eg in \cite{Barwise}).
Indeed,
if we consider the elements
of $\cO$ as \emph{sentences} (in a  language for first--order logic),
and for each $c \in  \cO$ the set $\lb c\rb\b$ as the class
of \emph{structures} (\ie models) in which
$c$ is true, then
$X \sco b$ states that
the class of structures in which
all $a$ in $X$ are true is a subclass
of the class of structures in which $b$
is true. In this sense,
 the definition of  $X \sco b$
 is  very similar in spirit  to the standard definition of the relation of  semantical consequence
in logic.

We now observe that the relation $\sco$ is indeed
a  generalization of  $\spo$.
\begin{proposition} \label{gen} For every $a$ and $b$ in $\cO$, we have \,\!\! $a  \spo b$ \  if and only if \
 $\lb a \rb \sco b$  \!\,.
\end{proposition}
\begin{proof}
Let  $a$ and $b$ in  $\cO$.
Since $\bigcap_{c \in \{ a \}}\lb c \rb\b = \lb a \rb\b$,
we have

\vspace{-0.35cm}
\begin{equation}\tag*{$\qedhere \square$}
\mbox{$
a \spo b  \quad $iff$ \quad  \lb a \rb\b\subseteq \lb b \rb\b
  \quad $iff$ \quad \bigcap_{c \in \{ a \}}\lb c \rb\b \subseteq \lb b \rb\b     \quad $iff$ \quad  \lb a \rb \sco b \enspace.$}
\end{equation}
\end{proof}

\begin{example}  Let $I$ be the triad given in Example \ref{ex triad}.
\sqi
\item  For every  $r \in \lb 0,1,2\rb$
and every $X \subseteq \cP_I$,
we have
 $ X  \, \LHD_{\cP_I} \, (r,P)\,$ \ifff either  $X = \lb (r,P) \rb$ or
$X$ contains at least two elements.

\item  Similarly,
for every  $s \in \lb 0,1,2\rb$
and every $X \subseteq \cN_I$,
we have
 $ X  \, \LHD_{\cN_I} \, (s,N)\,$ if and only if either  $X = \lb (s,N) \rb$ or
$X$ contains at least two elements. \hfill $\triangle$
\sqe
\end{example}

\begin{lemma}\label{int prop}
For every  $X \subseteq \cO$,
we have
 $ X \b= \bigcap_{a \in X} \lb a \rb\b$. \end{lemma}
\begin{proof}
Let   $X\subseteq \cO$, and let $b \in \overline{\cO}$. We have:

\vspace{-0.75cm}
\begin{equation}\tag*{$\qedhere \square$}
\mbox{$
b \in  \bigcap_{a \in X} \lb a \rb\b  \quad $iff$ \quad  b \in \lb a\rb\b$ \ for every $a \in X
  \quad $iff$ \quad \orth{a}{b}$ \ for every $a \in X    \quad $iff$ \quad  b \in X\b \enspace.$}
\end{equation}
\end{proof}

We  now characterize  the orthogonality relation $\bot$ in terms of the relations of semantical consequence.

\begin{proposition} \label{orth prop} Let  $a \in \cO$ and $b \in  \overline{\cO}$. Then, the following claims are equivalent:

\smallskip

{\centering

$
\begin{array}{rlcrlcrl}
\hbox{\emph{(1)}} &  \orth{a}{b} \!\ ;
 & \phantom{as} &
\hbox{\emph{(2)}}  & \lb b \rb\b \sco a \!\ ; &  \phantom{as} &
\hbox{\emph{(3)}}  & \lb a \rb\b \scd b\enspace.  \end{array}
$

\par}

\vspace{-0.2cm}
\begin{proof} (1) implies (2) : \ Assume $\orth a b$.
Then, we have $a \in \lb b\rb\b$, \ie
$\lb a \rb  \subseteq \lb b\rb\b$.
So, $\lb b \rb\b\b  \subseteq \lb a\rb\b$.
As
$\lb b \rb\b\b = (\lb b \rb\b)\b$, we conclude
$\lb b \rb\b \sco a $ by Lemma  \ref{int prop}.

 (2) implies (3) : \ Assume $\lb b \rb\b \sco a $. By   Lemma \ref{int prop}, we obtain
 $  \lb b \rb\b\b =  (\lb b \rb\b)\b \subseteq \lb a \rb\b$. Thus,  $ (\lb a \rb\b)\b =  \lb a \rb \b\b \subseteq \lb b \rb\b \b \b = \lb b \rb \b$. By Lemma again, \ref{int prop}, we finally get $\lb a \rb\b \scd b $.

  (3) implies (1) : \ Assume $\lb a \rb\b \scd b $. By Lemma \ref{int prop}, we get   $  \lb a \rb\b\b =  (\lb a \rb\b)\b \subseteq \lb b \rb\b$.
 Hence, $a \in \lb a\rb \subseteq\lb a \rb\b \b \subseteq \lb b \rb \b $.  Thus, we conclude
 $\orth a b $.
 \end{proof}

\end{proposition}

The following theorem is the generalization
of Corollary  \ref{coro sp}  which we are  looking for.

\begin{theorem} \label{thm fund} For every $X \subseteq \cO$ and every $b \in \cO$,  we have \,\!\!
$
b \in X\b\b
$ \  if and only if \    $X\sco b$ \!\,.

 \end{theorem}
\begin{proof} Let $X \subseteq \cO$ , and let $b \in \cO$. Suppose that $b \in X \b\b$.
Then, we have $\lb b\rb \subseteq X \b \b$.
Thus,   $X \b \b\b = X \b  \subseteq \lb  b\rb\b $. By  Lemma \ref{int prop}, we have $X\b = \bigcap_{a \in X}\lb a \rb\b$. Hence, we obtain
 $X \sco b$ .
 As for the converse, assume that $X \sco b$,
\ie $ \bigcap_{a \in X}\lb a \rb\b \subseteq \lb b \rb\b$.  By  Lemma \ref{int prop}, we have $X\b \subseteq \lb b \rb \b$. So, we obtain
$\lb b \rb\b\b   \subseteq X\b\b$. As $b \in \lb b \rb\b\b$, we conclude $b \in X \b\b$.
\end{proof}

\subsection{Entailment Systems} \label{ent sys sec}

We now introduce the notion of entailment system.
This notion can be seen as the natural adaptation of the concept
of  \emph{information system} \cite{Scott} (see also \cite{Barwise,Zhang})
to our setting.
We do not claim at all
that our concept of entailment system  constitutes   a novelty:
 structures of the same nature
 | often called  \emph{consequence relations} | has already been studied,
for different purposes,  in the  literature of abstract algebraic logic (see \eg\cite{AAL}).

In this paper,  we  introduce this notion in
order to show that the set $\cO$ equipped by the relation
of semantical consequence $\sco$
forms an entailment system
(Theorem \ref{thm ent}). One of the consequences of this fact is that
we can use the ``structural  rules" of Proposition \ref{prop struct} to derive properties of  terms.

\begin{definition}[Entailment system] \label{ent} We call \textbf{entailment systems}
any ordered pair  $(T,\Vdash)$ where:
\sqii
\item[$\bullet$] $T$ is a set. Its elements are said to be \textbf{tokens}, and we use $u,v,w,\ldots$ to range over them.
\item[$\bullet$]  $\Vdash \ \subseteq \pow(T) \times T$ is a binary relation such that
 for each
$u \in T$,  every $U \subseteq T$ and every $V \subseteq T$:
\sqii
\item[Axiom] \ : \qquad   $u \in U$ \ \  implies \ \  $U  \Vdash u$ \;\!;
\item[Cut] \ : \qquad    $U \Vdash v $ \,\! for every $ v  \in V$ \ and \ $V  \Vdash u$ \ \  imply \ \ $U  \Vdash u$ \ .
\sqee
We call $\Vdash$ the \textbf{entailment relation} of the entailment system and read ``$U \Vdash u$" as ``$U$ entails $u$.''

\sqee
Given $U \subseteq T$,  we call $ \lb u \in T  \st  U \Vdash u \rb$ the \textbf{set
of  consequences of $U$}.
\hfill $\triangle$
 \end{definition}

\begin{proposition}[Structural rules] \label{prop struct}
Let $(T,\Vdash)$ be an entailment system.
Let $U$ and $V$ be subsets of $T$, and let $u$ and $v$  be   tokens. Then,
we  have

\smallskip

{\centering
$
\begin{array}{rlcrl}
  \!\!\!\hbox{\emph{Axiom$_0$ :}} \!\!&\!\!  \mbox{$\lb u\rb \Vdash u$ \,\!;
}
 & \phantom{as} &
   \!\!\!\!\!\!\!\!\! \!\!\!\!\!\!\!\!\hbox{\emph{Weakening :}} \!\!&\!\! \mbox{$V \Vdash u$ \,\!\! and  \,\!\! $V \subseteq U$ \ imply \ $U \Vdash u$ \,\!;  }  \\   \!\!\!\!\hbox{\emph{Cut$_0$ :}} \!\!&\!\! \mbox{$U \Vdash v$  \,\!\!  and  \,\!\!   $U\cup \lb v\rb  \Vdash u$ \ imply \ $U  \Vdash u$  \,\!;  }  & &
    \!\!\!\!\!\!\!\!\!\!\!\!\!\!\!\!\!\!\!\hbox{\emph{Transitivity :}} \!\!&\!\! \mbox{$U \Vdash v$   \,\!\!  and   \,\!\!  $\lb  v \rb  \Vdash u$ \ imply \  $U  \Vdash u$\enspace.} \\

  \end{array}$
  \par
  }

\vspace{-0.2cm}

\end{proposition}
\begin{proof} Axiom$_0$ : \  We always have $u \in \lb u\rb$. Hence,  $\lb u\rb \Vdash u$ by Axiom.

Weakening : \  Suppose that $V \Vdash u$ and   $V\subseteq U$.
Let $v \in V$. Since $V \subseteq U$, we have $v \in U$. Hence, $U \Vdash v$ by Axiom. Since this holds for every $v \in V$, we obtain
$U \Vdash v$  for every $ v \in V$. Hence,
$U \Vdash u$ by  Cut.

Cut$_0$ : \ Suppose that $U \Vdash v$ and   $U \cup \lb v \rb \Vdash u$.
Let $w \in U \cup \lb v\rb$. If $w \in U$, then
 $U \Vdash w$
 by Axiom.
 If $w =v$, then
$U   \Vdash v$ by assumption. Hence, we have $U \Vdash w$ for every $ w \in U \cup \lb v\rb$.
By assumption,
$U \cup \lb v \rb \Vdash u$. Therefore,
$U  \Vdash u$ by Cut.

Transitivity : \  Suppose that $\lb v \rb  \Vdash u$. Then, we have
 $ U \cup \lb v \rb  \Vdash u$
 by Weakening. Since
  $ U   \Vdash v$ holds by assumption, we conclude
  $ U \Vdash u$ by  Cut$_0$.
\end{proof}

\begin{theorem} \label{thm ent} The pair $(\cO,\sco)$ is an entailment system. 
\end{theorem}
\begin{proof} We have to show that
conditions Axiom and Cut of Definition
\ref{ent} hold.
As for Axiom, note that we have $\bigcap_{a \in U}\lb a \rb\b \subseteq \lb a\rb\b$ for every $a \in U$.
As for Cut, assume that
 $\bigcap_{a \in U}\lb a \rb\b \subseteq \lb b\rb\b$ for every $b \in V$ and that $\bigcap_{b \in V}\lb b \rb\b \subseteq \lb c \rb\b$.
 Then, we have
  $\bigcap_{a \in U}\lb a \rb\b \subseteq \bigcap_{b \in V}\lb b \rb\b$.
Since $\bigcap_{b \in V}\lb b \rb\b \subseteq \lb c \rb\b$, we conclude
  $\bigcap_{a \in U}\lb a \rb\b \subseteq  \lb c \rb\b$.
\end{proof}

  By   Theorem \ref{thm fund}, we have   $ \lb b \in \cO \st X \sco b \rb = X \b\b$,  for every $X \subseteq \cO$. Since, by Theorem  \ref{thm ent}, the pair
$(\cO,\sco)$ is an entailment system, we conclude that
every closed set $X\b\b$ in $\cO$ can be precisely described as 
\emph{the set of  consequences of $X$}.
 Furthermore,  by  Theorem \ref{thm ent} again, we can use the ``structural rules" of Proposition \ref{prop struct} in the  entailment system $(\cO,\sco)$.
We  now use some of them   to derive some simple properties of terms.

\begin{proposition} \label{special} For every $X \subseteq \cO$  and every
$b
\in (\emptyset_\cO)\b\b$, we have $X \sco b$.
\end{proposition}
\begin{proof}
 By Theorem \ref{thm fund}, we have
 $b \in(\emptyset_\cO)\b\b$ if and only if $ \emptyset_\cO \sco b$. Thus, we can conclude
 $X \sco b$ by a simple application of Weakening.
 \end{proof}

We now show a property
which connects the relations $\sco$ and $\scd$ to the orthogonality relation
$\bot$.
\begin{proposition} \label{special 1}
Let $a$ and $a'$ be elements of $\cO$
, and let   $b$ and $b'$ be elements of $\overline{\cO}$. Suppose that \,\! $\orth{a}{b}$\,\!,\,\!
 $\lb a \rb \sco a'$ \,\!
and \,\! $\lb b \rb  \scd b' $ holds. Then,
 \,\! $\orth{a'}{b'}$ holds as well.
Graphically,

\medskip

{\centering
$
\begin{array}{ccc}
\lb a \rb & \sco &  a'   \\

 \bot & & \\
\lb b \rb  & \scd &  b'  \\
  \end{array}
  $ \qquad  implies \qquad
  $
\begin{array}{ccl}
\lb a \rb & \sco &  a'   \\

 \bot & &\!\bot \\
\lb b \rb  & \scd &  b' \enspace.  \\
  \end{array}
  $
  \par}
\end{proposition}
\begin{proof}
 Suppose that \,\! $\orth{a}{b}$\,\!,\,\!
 $\lb a \rb \sco a'$ \,\!
and \,\! $\lb b \rb  \scd b' $. Then, we have $\lb b \rb\b \sco a$ by Proposition \ref{orth prop}. From this and
$\lb a \rb \sco a'$ we obtain
$\lb b \rb\b \sco a'$ by Transitivity.
Since $\lb b \rb  \scd b' $ is equivalent to $\lb b \rb\b \subseteq \lb b' \rb\b$,
we obtain $\lb b' \rb\b \sco a'$ by Weakening. By
Proposition \ref{orth prop},
this means  $\orth{a'} {b'}$.
\end{proof}

\section{Functionals}
\label{functionals}

In this section, we introduce
the notion of functional in our setting
and study some properties of functionals.
In Subsection \ref{fun subsec}
we give the formal definition of functionals, and in Subsection
\ref{sub prop fun} we define and study the notions of
continuous functional and functional
which preserves the relation of semantical consequence.
Finally, in Subsection \ref{lin sec}
we introduce the notion of regular functional and show that regular functionals are continuous and
preserve the relation of semantical consequence.

\subsection{Functionals} \label{fun subsec}

We now define the concept of collection of functionals for a triad.
Recall that 
$ A= (\cP,\cN,\ort)$ is the arbitrary triad that we fixed at the beginning of Section \ref{closed sets}. Also, remember that 
$\dom(A)$, the domain of $A$, is the set $\cP \cup \cN$ (see Definition \ref{defss}).

\begin{definition}[Collection of functionals for a triad] \label{prefu def}
A \textbf{collection of functionals for} $A$ is an ordered pair $F = (\cF_F,\widehat{\phantom{a}}^F)$
where:
\squishlist
\item $\cF_F = \lb f,g,h ,\ldots \rb $ is a set. Its members  are called \textbf{functionals}.
\item $\widehat{\phantom{a}}^F$ is a function, that we call \textbf{interpretation}, which maps 
each functional  $ f \in \cF_F$ to a function $\widehat{f}^F$  from $\dom(A)$ to $\dom(A)$.
Furthermore, the function
$\widehat{\phantom{a}}^F$ has to satisfy
 the  following condition, called  \textbf{preservation of polarity}:

\vspace{-0.45cm}
\begin{equation}\tag*{$\triangle$}
\mbox{$\widehat{f}^F(p) \in \cP$  \ \  and \ \  $\widehat{f}^F(n) \in \cN$\enspace, \qquad for every $f \in \cF_F$,  every $p \in \cP$  and every $n \in \cN$\enspace.}
\end{equation}
\sqe
\end{definition}

\begin{example} \label{exfu} Let $I$ be the triad  given in Example \ref{ex triad}. Let $W \eqdef (\cF_W,\widehat{\phantom{a}}^W)$
be  the pair given by:
\squishlist
\item
$\cF_W \eqdef \lb \sharp, \flat, \natural \rb$, where  $\sharp$, $\flat$ and $\natural$
are just three pairwise distinct symbols.
\item   $\widehat{\phantom{a}}^W$ is the function which maps $\sharp$, $\flat$ and $\natural$ to the functions $\widehat{\sharp}^W$,
$\widehat{\flat}^W$ and $\widehat{\natural}^W$  
from $\dom(I)$ to $\dom(I)$ respectively
given as follows:

\smallskip

{\centering
$\widehat{\sharp}^W(a)  \eqdef  $ $\left\{
  \begin{array}{ll}
  (1,P) &  \hbox{if $a \in \cP_I$} \\
      (1,N) &  \hbox{if $a \in \cN_I$\enspace,}
\end{array}
\right.$ \quad 
$\widehat{\flat}^W(a)   \eqdef  $ $\left\{
  \begin{array}{ll}
  (1,P) &  \hbox{if $a \in \lb (0,P), (1,P)\rb $} \\
  (2,P) &  \hbox{if $a = (2,P)$} \\
      (1,N) \enspace&  \hbox{if $a \in \cN_I$\enspace,}
\end{array}
\right.$\quad
$\widehat{\natural}^W(a)   \eqdef  a$\enspace,

\par}

\noindent for $a \in \dom(I)$. Note that the function $\widehat{\phantom{a}}^W$ satisfies the condition of preservation of polarity.
\squishend 
According to our definition, the pair 
$W$ is a   collection of functionals for $I$.
\hfill $\triangle$
\end{example}
From now on, up to the end of the paper, we fix
 an arbitrary     collection of functionals
$F= (\cF_F,\widehat{\phantom{a}}^F)$   for $A$  and an arbitrary functional $f \in \cF_F$.
To ease notation, in the sequel
 we  write  $\cF$,  $\widehat{\phantom{a}}$
 and $\widehat{f}$  for $\cF_F$,  $\widehat{\phantom{a}}^F$ and $\widehat{f}^F$, respectively. Similarly, we write 
 $\widehat{\sharp}$,
$\widehat{\flat}$ and $\widehat{\natural}$
for  $\widehat{\sharp}^W$,
$\widehat{\flat}^W$ and $\widehat{\natural}^W$, respectively. 

Let us now discuss  Definition \ref{prefu def}.

 Intuitively, if we think of
the triad $A$  as a structure (\ie model) for a first--order language, then
\squishlist
\item 
$\cF$ can be seen as the set
of (unary) function symbols
of a first--order language;
\item
  $\widehat{\phantom{a}}$ can be seen as an interpretation of the function symbols in  the structure
 $A$   \ie as
 a function which
maps each
function symbol $f$ in $\cF$ to
a (unary) function $\widehat{f}$ from the domain of $A$ to itself.
\squishend
With this analogy in mind,
it is clear 
 that functionals are \emph{not} required to be functions.
For instance, in Example \ref{exfu},  the \emph{symbols} $\sharp$ and $\flat$ and $\natural$ are certainly not functions,
but  they are \emph{interpreted} in
the triad $I$ as the functions
$\widehat{\sharp}$, $\widehat{\flat}$
and $\widehat{\natural}$ 
from 
domain  of $I$
 to itself given above.

In this paper, we are not considering functionals because we want to form a category, say with $\dom(A)$ as the unique object  and with $\cF$ as the collection of morphisms (essentially, this is what is done in \cite{DBLP:journals/entcs/BasaldellaST10}).
In fact, 
the set  $\cF$ need not  contain any functional  intended to be
the identity morphism of $\dom(A)$. Also,  functionals are  not equipped with any  operation of  composition.
  In this article,  we  want to study functionals from a  different point of view. Namely, we want to  analyze their
relationship with the notions of
closed set,  continuity,    semantical consequence,
and   regularity.

The condition of preservation of polarity comes
 from ludics: in that setting, functionals (as defined in  \cite{DBLP:journals/entcs/BasaldellaST10})
always satisfy this property (see also Example \ref{fun lud}).
Except for this condition, note that
we do not  impose any  restriction on the nature 
of the interpretation function
 $\widehat{\phantom{a}}$. In particular,
it may happen that the interpretation  function
$\widehat{\phantom{a}}$  maps
 two \emph{distinct} functionals $g$ and $h$ in $\cF$ 
   to the
 \emph{same} function.

 \subsection{Properties of Functionals}
 \label{sub prop fun}

In this subsection, we relate functionals to closed sets, continuity,  and the  relation of semantical consequence.
To begin with,  it is convenient to  introduce some auxiliary notions and notation.

\begin{definition}[Image, pre--image]
Let
$X \subseteq \cO$. We call \textbf{image of $X$ under $f$} and \textbf{pre--image of $X$ under $f$} the subsets $f^\rightarrow(X)$
and $f^\leftarrow(X)$ of $\cO$ given by:

\smallskip

{\centering
$ f^\rightarrow(X) \ \eqdef \ \lb \widehat{f}(a) \st  a \in X   \rb$ \quad \  and \quad \
$ f^\leftarrow(X) \ \eqdef \  \lb a  \st \widehat{f}(a) \in X   \rb$ \enspace,
\par}
\smallskip
\noindent respectively. Equivalently, for
$a \in \cO$ we have
\vspace{-0.3cm}
\begin{equation}\tag*{$\triangle$}
\mbox{$
a \in f^\rightarrow(X) \ \reldef \ a = \widehat{f}(b)$   for some $b \in X$
\qquad \ and  \qquad \
$a \in f^\leftarrow(X) \ \reldef \ \widehat{f}(a) \in X\enspace.$}
\end{equation}
\end{definition}

\begin{example} \label{ex con}
 Let $I$ be the triad  given in Example \ref{ex triad}. Let $W$ be the collection of functionals for $I$ given in Example \ref{exfu}. We calculate the pre--images of some closed sets in
 $\cP_I$ and $\cN_I$ (\cf Example \ref{ex triad cont}).
\sqi
\item For $X \subseteq \cP_I$, we have

\vspace{-0.15cm}

{\centering
$\sharp^\leftarrow(X^{\bot_I\bot_I})$ \  = \
$\left\{
  \begin{array}{ll}
\emptyset_{\cP_I} &  \hbox{if  $X \in \lb \emptyset_{\cP_I}, \lb (0,P)\rb ,
\lb (2,P)\rb \rb$} \\
     \cP_I &  \hbox{otherwise\enspace.}
\end{array}
\right.$
 \par }

\noindent Similarly,  for $X \subseteq  \cN_I$, we have

\vspace{-0.15cm}

{\centering
$\sharp^\leftarrow(X^{\bot_I\bot_I})$   =
$\left\{
  \begin{array}{ll}
\emptyset_{\cN_I} &  \hbox{if  $X \in \lb \emptyset_{\cN_I}, \lb (0,N)\rb ,
\lb (2,N)\rb \rb$} \\
     \cN_I &  \hbox{otherwise\enspace.}
\end{array}
\right.$
\par }

\item We have $\flat^{\leftarrow}(\lb (1,P)\rb^{\bot_I\bot_I}) = \lb (0,P), (1,P)\rb $,  and
 $\flat^\leftarrow(X^{\bot_I\bot_I}) =
\sharp^\leftarrow(X^{\bot_I\bot_I})$ for every  $X \subseteq  \cN_I$.
\item We have
$\natural^{\leftarrow}(X ^{\bot_I\bot_I}) = X^{\bot_I\bot_I}$ for every $X \subseteq \cP_I$, and
$\natural^{\leftarrow}(X^{\bot_I\bot_I}) = X^{\bot_I\bot_I}$ for every $X \subseteq \cN_I$. \hfill $\triangle$
\sqe\end{example}

The following lemma  establishes some simple but fundamental  facts that we need in the sequel.

\begin{lemma} \label{lem image}
 Let
$X$ and $Y$ be subsets of $\cO$. Then, we have:

\smallskip

{\centering

$
\begin{array}{rlcrl}
\hbox{\emph{(1)}} &  \mbox{$X \subseteq Y$ \ implies \ $f^\rightarrow(X)  \subseteq  f^\rightarrow(Y)$ \!\ ;}
 & \phantom{as} &
\hbox{\emph{(2)}} &   \mbox{$X \subseteq Y$ \ implies \ $f^\leftarrow(X)  \subseteq  f^\leftarrow(Y)$ \!\ ;}   \\
\hbox{\emph{(3)}}  &\mbox{$ f^\rightarrow\big(f^\leftarrow(X)\big)\subseteq  X$ \!\ ;}  & \phantom{as} &
\hbox{\emph{(4)}} &  \mbox{$ X  \subseteq  f^\leftarrow\big(f^\rightarrow(X)\big)$\enspace.}
  \end{array}
$

\par}

\vspace{-0.2cm}

\end{lemma}
\begin{proof}
(1) : \  Suppose that $a \in f^\rightarrow(X)$, and assume $X \subseteq Y$. Then, $a = \widehat{f}(b)$ for some $b \in X$, by definition of image.
Since $X \subseteq Y$, we have
$b \in Y$.  So, $a = \widehat{f}(b)$ for some $b \in Y$.
Thus,   $a \in f^\rightarrow(Y)$.

(2) : \  Suppose that $a \in f^\leftarrow(X)$, and assume $X \subseteq Y$. Then, $ \widehat{f}(a) \in X$, by definition of pre--image.
Since $X \subseteq Y$, we have
$ \widehat{f}(a) \in Y$. Therefore,    $a \in f^\leftarrow(Y)$.

(3) : \  Suppose that $a \in f^\rightarrow\big(f^\leftarrow(X)\big)$. Then, $a = \widehat{f}(b)$ for some $b \in f^\leftarrow(X)$, by definition of image. Also, we have $\widehat{f}(b)  \in X$ by definition of pre--image. Hence, $a \in X$.

(4) : \  Suppose that $a \in X$. Then, $\widehat{f}(a)  \in f^\rightarrow(X)$  by definition of image.  Hence,
$a \in f^\leftarrow\big(f^\rightarrow(X)\big)$, by definition of pre--image.
 \end{proof}

We are now in position to define the notion of continuous functional.
We recall from \cite{closures} that
a \emph{continuous function} from a closure space
$(X,\Gamma)$ to a closure space $(Y,\Delta)$ (here $X,Y$ are sets and $\Gamma,\Delta$ are closure operators) is a function
$F$ from $X$ to $Y$
such that for every closed set $W$ in $(Y,\Delta)$  the pre--image of $W$ under $F$  is a  closed set in $(X,\Gamma)$.
In our setting, we define the concept
of continuity for functionals in a similar fashion.

\begin{definition}[Continuous functional]
 We say that
$f$ is \textbf{continuous in $\cO$} if
for every $X \subseteq \cO$ the set
$f^\leftarrow(X\b\b)$
  is  a closed set in $\cO$.
That is,
  \vspace{-0.25cm}
\begin{equation}\tag*{$\triangle$}
\mbox{$  f^\leftarrow(X\b\b) \b\b \ = \  f^\leftarrow(X\b\b)\enspace.$}
\end{equation}
\end{definition}

\begin{example} \label{ex fin} In the same notation
of Example \ref{ex con},  the following facts hold.
\sqi
\item The functional $\sharp$ is continuous in $\cP_I$
and $\cN_I$.

\item The functional $\flat$ is \emph{not} continuous
in $\cP_I$ (because $\flat^{\leftarrow}(\lb (1,P)\rb^{\bot_I\bot_I}) = \lb (0,P), (1,P)\rb$ and $\lb (0,P), $ \linebreak $ (1,P)\rb^{\bot_I\bot_I} = \cP_I \neq \lb (0,P), (1,P)\rb$).
On the other hand, the functional $\flat$ is continuous in $\cN_I$.

\item The functional $\natural$ is continuous in $\cP_I$
and $\cN_I$. \hfill $\triangle$
\sqe\end{example}

We now define the notion of
preservation of the relation of specialization for functionals.

\begin{definition}[Preservation of the relation of specialization $\spo$]
We say that
$f$ \textbf{preserves the relation of specialization  $\spo$} if
\vspace{-0.30cm}
\begin{equation}\tag*{$\triangle$}
\mbox{$a \spo b$  \ implies \ $\widehat{f}(a) \spo  \widehat{f}(b)$\enspace, \qquad for   every $a$ and $b$ in $\cO$\enspace.
}
\end{equation}

\end{definition}
 Generalizing  the previous notion, we   naturally obtain  the   definition of preservation  of the relation of semantical consequence.

\begin{definition}[Preservation of the relation of semantical consequence $\sco$]
 We say that
$f$ \textbf{preserves the relation of semantical consequence $\sco$} if

\vspace{-0.3cm}
\begin{equation}\tag*{$\triangle$}
\mbox{$X \sco b$  \ implies \ $f^\rightarrow(X) \sco  \widehat{f}(b)$\enspace, \qquad for   every $X \subseteq \cO$
and every $b \in \cO$\enspace.}
\end{equation}

\end{definition}

The following  theorem gives us some equivalent characterizations of the notion of continuity for functionals.
The most important one is the equivalence between (1)
and (4),
because it allows us to understand
the  concept  of continuity in $\cO$  as an
``inference rule" of the entailment system
$(\cO,\sco)$.

\begin{theorem}[Equivalent characterizations of continuity] \label{equiv}
The following statements are equivalent.

\smallskip

{\centering
$
\begin{array}{rlcrl}
\hbox{\emph{(1)}} &  \mbox{$f$ is continuous in $\cO$ \!\ ;}
 & \phantom{as} &
\hbox{\emph{(2)}} &   \mbox{$f^\leftarrow(X)\b\b \subseteq f^\leftarrow(X \b\b)$, \!\ for every $X \subseteq \cO$ \!\ ;}   \\
\hbox{\emph{(3)}}  &\mbox{$f^\rightarrow(X\b\b) \subseteq f^\rightarrow(X)\b\b$, \!\ for every $X \subseteq \cO$ \!\ ;}  & \phantom{as} &
\hbox{\emph{(4)}} &  \mbox{$f$ preserves  $\sco$\enspace.}
  \end{array}
$

\par}

\vspace{-0.2cm}

\end{theorem}
\begin{proof}
(1) implies (2) : \
 Let $X \subseteq \cO$.
  As $X \subseteq X\b\b$,
we have $f^\leftarrow(X) \subseteq f^\leftarrow(X\b\b)$, by Lemma \ref{lem image}(2). Hence,
 $f^\leftarrow(X) \b\b  \subseteq f^\leftarrow(X\b\b)\b\b$.
 Assume that $f$ is continuous in $\cO$.
We have $ f^\leftarrow(X\b\b)\b\b =  f^\leftarrow(X\b\b)$.
 Therefore,
 $f^\leftarrow(X)\b\b \subseteq f^\leftarrow(X \b\b)$.

 (2) implies (3) : \
 Let $X \subseteq \cO$.
 Let $Y \eqdef f^\rightarrow(X)$.
 By Lemma \ref{lem image}(4), we have
 $ X  \subseteq  f^\leftarrow\big(f^\rightarrow(X)\big) = f^\leftarrow(Y)$. So,
  $ X \b\b  \subseteq  f^\leftarrow(Y)\b\b$.
 Assume that (2) holds. We have
 $f^\leftarrow(Y)\b\b \subseteq f^\leftarrow(Y \b\b)$. Thus,
 $ X \b\b  \subseteq f^\leftarrow(Y \b\b)$.  By Lemma \ref{lem image}(1),
 we obtain
  $ f^\rightarrow(X \b\b)  \subseteq f^\rightarrow\big(f^\leftarrow(Y \b\b)\big)$.
  By Lemma  \ref{lem image}(3),
  we have $ f^\rightarrow\big(f^\leftarrow(Y \b\b)\big) \subseteq Y\b\b$.
 Therefore, $f^\rightarrow(X \b\b)  \subseteq  Y\b\b =  f^\rightarrow(X)\b\b$.

 (3) implies (4) : \
 Let $X \subseteq \cO$, and let
$b \in \cO$. Suppose that
$X \sco b$. By Theorem \ref{thm fund},
this means $b \in X\b\b$.
So,  we have  $\widehat{f}(b) \in f^\rightarrow(X\b\b) $ by  definition of image. Assume that (3) holds.
Then, we have
$ f^\rightarrow(X\b\b) \subseteq  f^\rightarrow(X)\b\b$.
Thus, $\widehat{f}(b) \in
f^\rightarrow(X)\b\b$.
By Theorem \ref{thm fund},
this is equivalent to  $f^\rightarrow(X) \sco  \widehat{f}(b)$.

 (4) implies (1) : \ Let $X \subseteq \cO$, and  let $Y\eqdef f^\leftarrow(X\b\b)$.
 We have to show that $Y \b\b = Y$. Clearly,
 $Y \subseteq Y\b\b$.
  To show the converse, let
 $b \in Y\b\b$. By  Theorem \ref{thm fund}, this means $ Y \sco b$.     Assume that (4) holds.
  Then, we have $f^\rightarrow(Y) \sco \widehat{f}(b)$.
  By  Theorem \ref{thm fund} again, this means $ \widehat{f}(b) \in f^\rightarrow(Y)\b\b$.
Hence, we have
   $\widehat{f}(b) \in f^\rightarrow(Y)\b\b
   = f^\rightarrow\big(f^\leftarrow(X\b\b)\big)\b\b \subseteq (X\b\b\big)\b\b = X\b\b$,
   by using Lemma \ref{lem image}(3).
 Since    $\widehat{f}(b) \in   X\b\b$,
we have $b \in f^\leftarrow(X\b\b)$
 by definition of pre--image.  Since
 $ f^\leftarrow(X\b\b) =  Y$, we conclude $Y\b\b \subseteq   Y$. This shows that $f$ is continuous in $\cO$.
\end{proof}

Analogously to what happens in the theory of closure spaces, in our setting
 we have  that
 continuous functionals
preserve the relation of
specialization. Before showing this, we now prove a simple lemma.

\begin{lemma} \label{lem im}
Let $a \in \cO$. Then, we have
$\lb \widehat{f}(a) \rb \   = \  f^\rightarrow(\lb a \rb)  $.
\end{lemma}
\begin{proof}
Let $b \in \cO$. We have, by using the definition of image:

\vspace{-0.70cm}
\begin{equation}\tag*{$\qedhere \square$}
\mbox{$
b \in \lb \widehat{f}(a) \rb   \quad $iff$ \quad b = \widehat{f}(a)
  \quad $iff$ \quad b = \widehat{f}(c)$ for some $c \in \lb a\rb    \quad $iff$ \quad  b \in f^\rightarrow(\lb a \rb) \enspace.$}
\end{equation}
\end{proof}

\begin{corollary} \label{coroll} Suppose that 
 $f$ is  continuous in $\cO$.
Then, $f$  preserves the relation of specialization $\spo$.
\end{corollary}

\begin{proof}
Suppose that $f $   is continuous in $\cO$. Then,
by Theorem \ref{equiv}((1) implies (4)), the functional  $f$ preserves the relation of semantical consequence $\sco$. Let $a$ and $b$ in $\cO$, and suppose that
$a \spo b$. By Proposition \ref{gen},   $ a  \spo b$ is equivalent to  $\lb a \rb \sco b$.
So,  we obtain  $f^\rightarrow(\lb a \rb) \sco \widehat{f}(b)$
by preservation of $\sco$.
By Lemma \ref{lem im}, we have $f^\rightarrow(\lb a \rb) = \lb \widehat{f}(a) \rb $. Therefore,
 $\lb \widehat{f}(a) \rb \sco \widehat{f}(b)$. The latter  is equivalent to
$\widehat{f}(a) \spo \widehat{f}(b)$, by using Proposition \ref{gen} again.
\end{proof}

\subsection{Regularity} \label{lin sec}

We now introduce the concept
of regular  functional.
The reason for introducing this concept comes from ludics:
 in that setting every functional (in the sense of \cite{DBLP:journals/entcs/BasaldellaST10}) is regular  (see  Example \ref{fun lud}).

\begin{definition}[Regular  functional] \label{fu def}
We say that $f$ is  \textbf{regular} if the following condition holds:

\vspace{-0.33cm}
\begin{equation}\tag*{$\triangle$}
\mbox{$\widehat{f}(p) \, \ort \, n$ \ \  if and only if \,    $p \, \ort \, \widehat{f}(n)$\enspace,\qquad for every  $p \in \cP$  and every $n \in \cN$\enspace.}
\end{equation}
 \end{definition}

We observe that our condition of regularity is   analogous to the condition  of  \emph{linearity} for maps in \cite{DBLP:conf/lics/LafontS91}. There,
maps between \emph{games} are said to be \emph{linear}
if they satisfies a similar condition (see Example \ref{fun chu}).
Unfortunately, the adjective
``linear" is already present in the vocabulary of ludics
 \cite{DBLP:journals/tcs/Terui11,DBLP:journals/entcs/BasaldellaST10},
and  it denotes a property of designs which has nothing to do with the condition above.  To avoid any sort of confusion, we decided to introduce a different  terminology.
 We also remark that  in the standard terminology for Chu spaces, the condition of linearity of  \cite{DBLP:conf/lics/LafontS91}  is commonly called \emph{adjunction} condition.

We now show some  equivalent characterization of the notion of regularity.
Before doing this, it is convenient to introduce some now terminology.

\begin{definition}[Various properties of functionals]
We say that:

\sqi
\item $f$ is  \textbf{semiregular in $\cO$}
\,\! if \,\!  $\widehat{f}(a) \, \ort \, b$ \ \ implies \ \   $a \, \ort \, \widehat{f}(b)$ \,\!, \ for every  $a \in \cO$  and every $b \in \overline{\cO} \!\;; $
\item $f$ is  \textbf{$\rightarrow \, \leftarrow$ in $\cO$} \,\! if \,\! $ f^\rightarrow(X)\b \subseteq f^\leftarrow(X\b) $ \,\!,  \!\,
 for every $X \subseteq \cO$ \!\,;
 \item $f$ is  \textbf{$\leftarrow \, \rightarrow$ in $\cO$} \,\! if \,\! $   f^\leftarrow(X\b) \subseteq f^\rightarrow(X)\b $ \,\!,  \!\,
 for every $X \subseteq \cO$ \!\,;
  \item $f$ is  \textbf{good in $\cO$} \;\!\!\;\! if \,\! $ f^\rightarrow(X)\b =  f^\leftarrow(X\b) $ \,\!,  \!\,
 for every $X \subseteq \cO$ \!\,. \hfill $\triangle$

\sqe

 \end{definition}

Note that we  have the following equivalences:

\medskip

{\centering
$
\begin{array}{lcl}
  \mbox{ $f$ is  regular}
 & \mbox{if and only if} &
\hbox{$f$
is semiregular in $\cO$ and semiregular in $\overline{\cO} \enspace; $} \smallskip
\vspace{-0.05cm}
\\
  \mbox{ $f$ is good in $\cO$}
& \mbox{if and only if} &
\hbox{$f$ is
$\rightarrow \, \leftarrow$ in $\cO$ and   $ \leftarrow \, \rightarrow \,$ in $\cO$\enspace.}
  \end{array}
$

\par }

 \begin{proposition} \label{prop prop}
The following claims are equivalent.

\smallskip

{\centering
$
\begin{array}{rlcrlcrl}
\hbox{\emph{(1)}} &  \mbox{$f$ is semiregular in $\cO$ \!\,;}
 & \phantom{as} &
\hbox{\emph{(2)}} &   \mbox{$f$ is $\rightarrow \, \leftarrow$ in $\cO$ \!\,;}
 & \phantom{as} &
\hbox{\emph{(3)}} &  \mbox{$f$ is
$ \leftarrow \, \rightarrow $
in $\overline{\cO}$\enspace.}
  \end{array}
$

\par}

\vspace{-0.2cm}

\end{proposition}
\begin{proof}
(1) implies (2) : \ Let $X \subseteq \cO$, and let $b \in \overline{\cO}$. Assume that  $b \in f^\rightarrow(X)\b$. Then, $\orth{c}{b}$ for every $c \in f^\rightarrow(X)$, by definition of orthogonal set.
So, we have
 $  \orth{\widehat{f}(a)}{b}$   for every $a \in X$
by definition of image.
 Assume (1). We obtain
  $  \orth{a}{\widehat{f}(b)}$   for every $a \in X$. Hence, $\widehat{f}(b) \in X\b$
  by definition of orthogonal set.
  To conclude, we get $b \in
 f^\leftarrow(X\b)$ by definition of pre--image.

(2) implies (3) : \ Let $Y \in \overline{\cO}$, and let $a \in \cO$.
Assume $a \in f^\leftarrow(Y\b)$.
We have $\widehat{f}(a) \in Y\b$ by definition of pre--image.
Thus, $\lb \widehat{f}(a) \rb \subseteq Y\b$.  Hence,
$f^\rightarrow(\lb a \rb) \subseteq Y\b$ by Lemma \ref{lem im}.
So, $ Y\b \b \subseteq f^\rightarrow(\lb a \rb) \b$. Since
$Y \subseteq Y\b\b$, we obtain $Y \subseteq  f^\rightarrow(\lb a \rb) \b$.
Assume (2). Since $\lb a \rb \subseteq \cO$, we have $ f^\rightarrow(\lb a \rb) \b \subseteq  f^\leftarrow(\lb a \rb\b)$. So,
$Y \subseteq  f^\leftarrow(\lb a \rb\b)$.  By Lemma
\ref{lem image}(1) and (3), we have $ f^\rightarrow(Y) \subseteq f^\rightarrow\big(f^\leftarrow(\lb a \rb\b)\big) \subseteq \lb a \rb \b$.
Thus, $ f^\rightarrow(Y) \subseteq \lb a \rb\b$.
Hence,  $ \lb a \rb\b\b \subseteq f^\rightarrow(Y) \b  $.
As $a \in \lb a \rb \b\b$, we  conclude $a  \in f^\rightarrow(Y) \b$.

(3) implies (1) : \  Let $a\in \cO$, and let  $b \in \overline{\cO}$.
Assume that $\orth{\widehat{f}(a)}{b}$  holds, \ie $\widehat{f}(a)   \in \lb b\rb\b$.  From  this, we obtain
$ \lb \widehat{f}(a) \rb \subseteq \lb b \rb\b$.
By Lemma \ref{lem im}, we have
$ f^\rightarrow(\lb a \rb ) \subseteq \lb b \rb\b$.
By Lemma
\ref{lem image}(2) and (4), we have $ \lb a \rb \subseteq f^\leftarrow\big(f^\rightarrow(\lb a \rb ) \big) \subseteq f^\leftarrow\big(\lb b \rb\b\big)$.
So, $\lb a \rb \subseteq f^\leftarrow\big(\lb b \rb\b\big)$.
Assume (3).
As $\lb b \rb \subseteq \overline{\cO}$, we obtain
$f^\leftarrow\big(\lb b \rb\b\big) \subseteq
f^\rightarrow(\lb b \rb)\b$.
Thus, $ \lb a \rb \subseteq f^\rightarrow(\lb b \rb)\b$.
By Lemma \ref{lem im}, we have
$  f^\rightarrow(\lb b \rb) = \lb \widehat{f}(b) \rb$.  Hence, $ \lb a \rb\subseteq  \lb \widehat{f}(b) \rb\b$. Thus, we get  $a \in
 \lb \widehat{f}(b) \rb\b$, \ie   $\orth{a}{\widehat{f}(b)}$.
\end{proof}

 \begin{theorem}[Equivalent characterizations of regularity] \label{reg car}
The following statements are equivalent.

\smallskip

{\centering
$
\begin{array}{rlcrlcrl}
\hbox{\emph{(i)}} &  \mbox{$f$ is regular \!\,;}
 & \phantom{as} &
\hbox{\emph{(ii)}} &   \mbox{$f$ is good in $\cO$ \!\,;}
 & \phantom{as} &
\hbox{\emph{(iii)}} &  \mbox{$f$ is
good
in $\overline{\cO}$\enspace.}
  \end{array}
$

\par}

\vspace{-0.2cm}

\end{theorem}
\begin{proof}

(i) implies (ii) : \ Assume that $f$ is regular. Then, as $f$ is semiregular in $\cO$,
it follows that $f$ is
$\rightarrow \, \leftarrow$ in $\cO$
by
  Proposition \ref{prop prop}((1) implies (2)).  As $f$   is also semiregular in $\overline{\cO}$,
we have that $f$ is
$ \leftarrow \, \rightarrow$ in $\cO$
by
  Proposition \ref{prop prop}((1) implies (3)). As a consequence of this,  $f$ is good in $\cO$.

  (ii) implies (iii) : \ Suppose that $f$ is good in $\cO$. Then, since $f$ is $\rightarrow \, \leftarrow$ in $\cO$,
  we have that $f$ is
$ \leftarrow \, \rightarrow$ in $\overline{\cO}$
by
  Proposition \ref{prop prop}((2) implies (3)). Similarly, as $f$ is
$ \leftarrow \, \rightarrow$ in $\cO$,
we have that $f$ is  $\rightarrow \, \leftarrow$ in $\overline{\cO}$
by
  Proposition \ref{prop prop}((3) implies (2)). Therefore, $f$ is good in $\overline{\cO}$.

(iii) implies (i) : \
Finally, assume that $f$ is good in $\overline{\cO}$.
Then, as $f$ is  $\rightarrow \, \leftarrow$ in $\overline{\cO}$, we have that $f$ is semiregular in $\overline{\cO}$,
by  Proposition \ref{prop prop}((2) implies (1)). Analogously,
since $f$ is  $\leftarrow \, \rightarrow $ in $\overline{\cO}$, we have that $f$ is semiregular in $\cO$,
by  Proposition \ref{prop prop}((3) implies (1)). This shows that $f$ is regular.
\end{proof}

 We now show the main result of this section:
  regular functionals
are continuous
and preserve the relation of semantical consequence.

\begin{theorem}[Regularity] \label{cont}  Suppose that $f$ is a regular functional. Then,
\smallskip

{\centering

$f$  is continuous in  \!\,\! $\cP$ \!\,\! and
\!\,\! $\cN$ \!\,\! and   preserves the relations of semantical consequence $\, \LHD_{\cP} $
\!\,\! and \!\,\!  $\, \LHD_{\cN}$\enspace.

 \par}

 \vspace{-0.2cm}

\end{theorem}

\begin{proof}
Suppose that $f$ is regular. Let
    $X \subseteq \cO$, and  let $Y \eqdef X\b$. By Theorem \ref{reg car}((i) implies (iii)), we have that $f$ is good in $\overline{\cO}$. So,
$f^\rightarrow(Y)\b =   f^\leftarrow(Y\b)$.
Thus,  $  f^\leftarrow(X\b\b) \b\b = f^\leftarrow(Y\b)\b\b =
( f^\rightarrow(Y)\b)\b\b =
 f^\rightarrow(Y)\b =
f^\leftarrow(Y\b) =
 f^\leftarrow(X\b\b) $.
This show that  $f$
is continuous in $\cO$.
Now,
by Theorem \ref{equiv}((1) implies (4)) we obtain that $f$ preserves the relation of semantical consequence $\, \LHD_{\cO} \, $.
\end{proof}

Finally, we   observe that
regularity is a concept which is stronger than continuity. Namely, we show
that continuity
in  $\cP$ \emph{and} $\cN$ does \emph{not} implies regularity in general. See Example \ref{ex fin1}(a) below.

\begin{example} \label{ex fin1} In the same notation
of Example \ref{ex con} and Example \ref{ex fin}, we have:
\sqi
\item[(a)] The functional $\sharp$ is continuous in $\cP_I$
\emph{and} $\cN_I$. However,
as $\widehat{\sharp}((0,P)) = (1,P) $ and
$\widehat{\sharp}((1,N)) = (1,N) $,  we have
$\, \widehat{\sharp}((0,P)) \, \bot_I \,  (1,N) \,$ and
$\, (0,P) \, \not\!\!\!\bot_I  \,  \widehat{\sharp}((1,N)) \,$. Therefore, the functional
$\sharp$ is \emph{not} regular.

\item[(b)] The functional $\flat$ is not continuous
in $\cP_I$ and hence, by Theorem \ref{cont}, it   cannot be regular.

\item[(c)] The functional $\natural$ is continuous in $\cP_I$
and $\cN_I$. It  is also regular,
as we have
\smallskip

{
\centering{
$ \, \widehat{\natural}((r,P)) \, \bot_I \,  (s,N) \,$ \ \ \ iff \ \ \
$ \, (r,P) \, \bot_I \,  (s,N) \,$ \ \ \  iff \ \ \ $ \, (r,P) \, \bot_I \,  \widehat{\natural}((s,N))$\enspace,}
\par }

\smallskip

\noindent for every $(r,P) \in \cP_I$
and every $(s,N) \in \cN_I$.  \hfill $\triangle$
\sqe\end{example}

\section{Examples: Boolean--Valued Games and Ludics} \label{example triads}

 We now give two ``abstract" examples of triad.

\begin{example}[Boolean--valued games] \label{games}  A \emph{Boolean--valued game} \cite{DBLP:conf/lics/LafontS91} (or \emph{Boolean--valued Chu space} \cite{DBLP:conf/lics/Pratt95}) can be presented  as an ordered  triple
$\sz = (\sP, \sO,\sr)$,  where $\sP$ (``strategies", ``points") and $\sO$ (``co--strategies", ``open sets")
are sets, and $\sr$ is a subset of $\sP \times \sO$, \ie a relation
from $\sP$ to $\sO$. Given
 $x \in \sP$ and $y \in \sO$,
 we also write $\, x \, \sr \, y\,$ for
$(x,y) \in \sr$ in the sequel.

Every  Boolean--valued game with $\sP$
and $\sO$  disjoint is a triad in our sense.
 \hfill $\triangle$
 \end{example}

\begin{example}[Ludics] \label{ludic}
 We now  show that ludics fits into our general setting. For convenience,
 we consider  ludics as formulated  in  \cite{DBLP:journals/tcs/Terui11}. In
 Subsection 2.1 of \cite{DBLP:journals/entcs/BasaldellaST10} the reader can find all the notions required to understand this example and Example \ref{fun lud}.
 Let $ \cA$ be  a   signature in the sense of  \cite{DBLP:journals/tcs/Terui11,DBLP:journals/entcs/BasaldellaST10}. Consider
the triple $\bA \eqdef (\cP_\bA,\cN_\bA, \ort_\bA)$
 where:
 \squishlist
 \item $\cP_\bA$ is the set of all linear, cut--free and positive designs with at most $x_0$ as free variable (\ie the set of the positive atomic designs
of \cite{DBLP:journals/entcs/BasaldellaST10} \emph{augmented by} $\Omega$)
over the signature
 $\cA$.

\item $\cN_\bA$ is the set of all linear, cut--free and negative designs without  free variables (\ie the set of the negative atomic designs of \cite{DBLP:journals/entcs/BasaldellaST10})
over the signature
 $\cA$.

 \item For $p \in  \cP_\bA$ and  $n \in \cN_\bA$, we set

{
\centering{
$p \, \ort_\bA \, n \ \reldef \  \ [\![ \; p[n/x_0] \; ]\!] = \maltese$\enspace,}

\par }

\noindent   where by   ``$[\![ \;  \; ]\!]$" we denote    the  \emph{normal form function} (see \cite{DBLP:journals/tcs/Terui11,DBLP:journals/entcs/BasaldellaST10}).
This orthogonality relation corresponds to the original orthogonality relation of ludics.
 \squishend
Since in ludics
the sets
 $ \cP_\bA$ and   $
  \cN_\bA$ are  disjoint, the triple  $\bA$ is a triad in our sense.
(In this example, $\Omega$ is a member of $\cP_\bA$.  In ludics, this situation is usually not allowed. The only differences | \wrt more traditional presentations of ludics | are the following:
  (i) the set $\emptyset_{\cN_\bA}$ is a closed set in $\cN_\bA$, because $\emptyset_{\cN_\bA}= \lb \Omega \rb^{\bot_\bA} $;
  (ii)
 the set  $\cP_\bA$  is the unique closed set in $\cP_\bA$ which contains $\Omega$, because $\lb \Omega \rb^{\bot_\bA\bot_\bA} = (\emptyset_{\cN_\bA})^{\bot_\bA} = \cP_\bA$ and hence, $\lb \Omega \rb \subseteq X^{\bot_\bA\bot_\bA}$ implies $\cP_\bA  = X^{\bot_\bA\bot_\bA}$.)
 \hfill $\triangle$ \end{example}

We now give examples of collections of functionals for the triads given in the previous two examples.

\begin{example}[Linear maps] \label{fun chu}
Let  $\sz = (\sP,  \sO,\sr)$ be a Boolean--valued game. A  \emph{linear map from $\sz$ to itself} is an ordered  pair $(\ssf, \ssg)$ of functions $\ssf : \sP \arr \sP$ and $\ssg: \sO \arr \sO$ such that

 \smallskip

{
\centering{
$\ssf(x)\, \ \sr \,\ y $  \ \ \ \ \ \   if and only if \ \ \ \  \ \ \ $\, x \,\ \sr \,\ \ssg(y) $\enspace,}

\par }

 \smallskip

\noindent
for every $x \in \sP$ and every $y \in \sO$.
(Here we are following the terminology introduced in
\cite{DBLP:conf/lics/LafontS91}.
We also mention that  in the standard terminology for Chu spaces, linear maps are also called \emph{Chu transforms}.)

Let $\sz = (\sP,  \sO, \sr)$ be a Boolean--valued game
with $\sP$ and $\sO$ disjoint. Recall that $\dom(\sz) = \sP \cup \sO$.
Let $L \eqdef (\cF_L,\widehat{\phantom{a}}^L)$, where:
\sqi\item
$\cF_L$ is  \emph{any subset} of  the set of all linear maps
from $\sz$ to itself.
\item $\widehat{\phantom{a}}^L$  maps each functional
  $(\ssf,\ssg) \in \cF_L$ to the function $\widehat{(\ssf,\ssg)}^L$
 from $\dom(\sz)$ to $\dom(\sz)$ given by:

\smallskip

{\centering
$\widehat{(\ssf,\ssg)}^L\!\!\!(z)   \ \eqdef \ $ $\left\{
  \begin{array}{ll}
\ssf(z) &  \hbox{if $z \in \sP$} \\
\ssg(z) &  \hbox{if $z \in\sO$\enspace, \quad    for   $z \in \dom(\sz)$\enspace.}
\end{array}
\right.$
\par }

\smallskip

  \sqe
We claim that \emph{$L$ is a collection of functionals for $\sz$} and that \emph{all functionals
in $\cF_L$ are regular}.  To prove these claims, let
$(\ssf,\ssg) \in \cF_L$,  $x \in \sP$
and $y \in \sO$.

By definition, we have  $\widehat{(\ssf,\ssg)}^L\!\!\!(x)   = \ssf(x) \in \sP$  and
$\widehat{(\ssf,\ssg)}^L\!\!\!(y)   = \ssg(y) \in \sO$, and this shows that preservation of polarity holds. 
As for regularity, we have that

\vspace{-0.75cm}

{\centering
\begin{equation}\tag*{$\triangle$}
\mbox{$\widehat{(\ssf,\ssg)}^L\!\!\!(x) \,\ \sr \,\ y\,$  \ \ \  iff \ \  \
$\, \ssf(x)\,\ \sr \,\ y \, $ \ \ \  iff \ \ \  
  $\, x \,\ \sr \,\ \ssg(y) \, $  \ \  \ iff \ \ \   $\, x \,\  \sr \,\  \widehat{(\ssf,\ssg)}^L\!\!\!(y)\,\enspace.$}
\end{equation}
\par}
\end{example}

\begin{example}[Functionals in ludics] \label{fun lud}Let $\bA$ be the triad defined in Example \ref{ludic}. 
Recall that $\dom(\bA) = \cP_\bA \cup \cN_\bA$. Consider the pair $H = (\cF_H,\widehat{\phantom{a}}^H)$, where:
\sqi\item
$\cF_H$
 is  \emph{any subset} of  the set of all linear, cut--free and negative designs with at most $x_0$ as   free variable (\ie the set of  \emph{functionals},   in the sense of  Subsection 2.2 of  \cite{DBLP:journals/entcs/BasaldellaST10})
over the signature $\cA$.

\item $\widehat{\phantom{a}}^H$
 maps each functional
  $g \in \cF_H$ to the function $\widehat{g}^H$
 from $\dom(\bA)$ to $\dom(\bA)$ given by:

\smallskip

{\centering
$\widehat{g}^H\!(a)   \ \eqdef \ $ $\left\{
  \begin{array}{ll}
{[\![ \; a[g/x_0] \; ]\!]} &  \hbox{if $a \in \cP_\bA$} \\
{[\![ \; g[a/x_0] \; ]\!]} &  \hbox{if $a \in\cN_\bA$\enspace, \quad for  $a\in \dom(\bA)$\enspace.
}
\end{array}
\right.$
\par }

\smallskip

  \sqe
We now  claim  that $H$ \emph{is a collection of functionals for $\bA$} and that \emph{all
functionals in $\cF_H$ are regular}.
To show these claims,  let $g \in \cF_H$,  $p \in \cP_\bA$ and
 $n \in \cP_\bA$.

Since    $\cP_\bA$ and $\cN_\bA$ are disjoint, it follows from 
the definition of the normal form function  that $\widehat{g}^H$
is  a well--defined function from $\dom(\bA)$ to itself.
 (It would be only a \emph{partial} function if $\Omega \notin \cP_\bA$,
as there are  many $q \in \cP_\bA$ | different from $\Omega$ |  and $h \in \cF_H$ | in case $\cF_H$ is the set of all functionals of \cite{DBLP:journals/entcs/BasaldellaST10} | such that $\![\![ \; q[h/x_0] \; ]\!] = \Omega$. Exactly for
this reason, we  included $\Omega$ in $\cP_\bA$.)
Furthermore,   $\widehat{g}^H\!(p)  \in \cP_\bA$  and  $\widehat{g}^H\!(n) \in \cN_\bA$
again follow from the definition of normal form function.
 This shows that preservation of polarity holds. 
As for regularity, this property  is a consequence of the   \emph{associativity}
of normalization (see \eg \cite{DBLP:journals/entcs/BasaldellaST10}).
(This fact has also been observed, without  proof,  in
\cite{DBLP:journals/entcs/BasaldellaST10}: see Lemma 2.4 and Equation (1) in  Section 5 of \cite{DBLP:journals/entcs/BasaldellaST10}.)
Indeed, we have:

\smallskip

{ \centering
$\begin{array}{rclr}
\widehat{g}^H\!(p) \, \ort_\bA \, n   & \mbox{iff}  &  [\![ \; [\![ \; p[g/x_0] \; ]\!] \; [n/x_0] \; ]\!]   =  \maltese    & \mbox{(by definition of $\bot_\bA$ and $\widehat{\phantom{a}}^H$)} \\
 & \mbox{iff}  & [\![ \; [\![ \; p[g/x_0] \; ]\!] [\; [\![ \; n \; ]\!]/x_0] \; ]\!]   =  \maltese   & \mbox{(because $n$ is cut--free)} \\
 & \mbox{iff}  & [\![ \;  p[g/x_0] [n /x_0] \; ]\!]   =  \maltese   & \mbox{(by associativity)} \\
 &\mbox{iff}   &  [\![ \;  p[g[n /x_0]/x_0] \; ]\!]   =  \maltese &   \mbox{(by substitution)} \\
  &\mbox{iff}   &  [\![ \;   [\![ \;p \; ]\!] [ \; [\![ \;g[n /x_0] \; ]\!]/x_0] \; ]\!]   =  \maltese   &   \mbox{(by associativity)} \\
   &\mbox{iff}   &  [\![ \;  p  [ \; [\![ \;g[n /x_0] \; ]\!]/x_0] \; ]\!]   =   \maltese &   \mbox{(because $p$ is cut--free)} \\
      &\mbox{iff}   &   p \, \ort_\bA \, \widehat{g}^H\!(n)\enspace.   &   \mbox{(by definition of $\bot_\bA$ and  $\widehat{\phantom{a}}^H$)}
\end{array}$

\vspace{-0.435cm}

\hfill $\triangle$

\par}

\end{example}

\section{Conclusion} \label{conclusions}

In this paper, we introduced the notion of triad in order to study, analyze, discover and  rediscover some properties
which hold in ludics  from a more abstract and general perspective.

In particular,   by applying of our abstract results to the concrete setting of ludics we arrive at the following conclusion.

\begin{theorem}[Abstract results on triads and functionals applied to ludics] \label{thm fin} In the notation  and terminology of Example \emph{\ref{ludic}} and
Example \emph{\ref{fun lud}}, we have:

\sqi
\item The pairs $(\cP_\bA\, , \, \LHD_{\cP_\bA}\, ) $ and $(\cN_\bA \, , \, \LHD_{\cN_\bA}\,)$
are entailment systems;
\item Functionals in ludics  are regular and therefore:
\vspace{-0.1cm}
\sqqi
\item[$\bullet$] Functionals in ludics
are continuous in $\cP_\bA$ and $\cN_\bA$;
\item[$\bullet$] Functionals in ludics preserve the relations of semantical consequence
$\, \LHD_{\cP_\bA}\, $ and $\, \LHD_{\cN_\bA}\, $ and thus:
\sqqi
\vspace{-0.6cm}
\item[$\bullet$]  Functionals in ludics preserve the relations of specialization
$\, \lhd_{\cP_\bA}\, $ and $\, \lhd_{\cN_\bA}\, $.
\sqqe
\sqqe
   \sqe
\end{theorem}
\begin{proof} By
Example \ref{ludic},
Example \ref{fun lud}, Theorem \ref{thm ent} , Theorem \ref{cont} and Corollary \ref{coroll}.
\end{proof}

To conclude the paper, we observe that from the point of view of ludics
our paradigmatic vision of \emph{designs as terms} is somehow limited because
here terms only correspond to \emph{atomic} designs. Even though  terms  capture
the most important class of designs (in the opinion of the present author), in ludics there
are  plenty of non--atomic designs which  do not fit into our framework.
For future work, we plan to extend our setting in order to cope with them.

\nocite{*}
\bibliographystyle{eptcs}
\bibliography{basald}

\end{document}